\documentclass[a4paper,11pt]{amsart}
\usepackage[utf8]{inputenc}
\usepackage[T1]{fontenc}

\usepackage{multicol}
\usepackage{mathpartir}
\usepackage{amssymb}

\title{Propositional dynamic logic for searching games with errors}

\author{Bruno \textsc{Teheux}}
\address{Mathematics Research Unit, FSTC, University of Luxembourg, 6, Rue Coudenhove-Kalergi, L-1359 Luxembourg, Luxembourg}
\email{bruno.teheux@uni.lu}
\thanks{The author was supported  by the internal research project F1R-MTHPUL-
12RDO2 of the University of Luxembourg.}


\newcommand{\Prop}{\mathsf{Prop}}
\newcommand{\Form}{\mathsf{Form}}
\newcommand{\model}[1]{\mathcal{#1}}
\renewcommand{\implies}{\rightarrow}
\newcommand{\struc}[1]{\langle #1 \rangle}
\newcommand{\Val}{\mathrm{Val}}
\newcommand{\lucas}{\mbox{\L}}
\newcommand{\logic}[1]{\mathsf{{#1}}}

\newcommand{\proves}{\vdash}

\newcommand{\var}[1]{\mathcal{#1}}

\newcommand{\free}{\mathcal{F}}

\newcommand{\card}[1]{\vert #1 \vert}

\newcommand{\poss}[1]{\langle #1 \rangle}
\newcommand{\necess}[1]{[#1]}
\newcommand{\fish}[1]{\mathrm{FL}(#1)} 

\newtheorem{thm}{Theorem}[section]
\newtheorem{cor}[thm]{Corollary}
\newtheorem{lem}[thm]{Lemma}
\newtheorem{fact}[thm]{Fact}
\newtheorem{prop}[thm]{Proposition}
\theoremstyle{definition}
\newtheorem{defn}[thm]{Definition}
\theoremstyle{remark}
\newtheorem{rem}[thm]{Remark}
\newtheorem{ex}[thm]{Example}
\numberwithin{equation}{section}

\begin{document}
\begin{abstract}
We investigate some finitely-valued generalizations of pro\-po\-sitional dynamic logic with tests. We start by introducing the $n+1$-valued \textsc{Kripke} models and a corresponding language based on a modal extension of \textsc{\L ukasiewicz} many-valued logic. 
We illustrate the definitions by providing a framework for an analysis of the  \textsc{Rényi}~-~\textsc{Ulam} searching game with errors.

Our main result is the axiomatization of the theory of the $n+1$-valued  \textsc{Kripke} models.  This result is obtained through filtration of the canonical model of the smallest $n+1$-valued propositional dynamic logic.  
\end{abstract}
\maketitle

\section{Introduction}
Propositional dynamic logic ($\logic{PDL}$) is  a multi-modal logic designed to reason about programs.  The general idea behind the semantic of this system  is the following.  Program states are gathered in a set $W$. Any program $\alpha$ is encoded by its input/output relation on $W$. Programs are built from atomic ones and test operators using regular operations.
Their associated relations are defined as to respect this algebra of programs. The goal is to provide a framework for formal verification through input/output specifications.

Since its introduction by  \textsc{Fischer} and \textsc{Ladner} in \cite{Ladner1979}, the scope of dynamic logic has widened to many other areas such as game theory (see \cite{Fitting2011, Harrenstein2003, Pauly2007}), epistemic logic (see \cite{vanDitmarsch2007}) or natural language (see \cite{vanBenthem2010}). The subject is under constant and active development (see \cite{Benevides2011, Benevides2010,  Doberkat2012, Leivant2008} for example) and we refer to \cite{Harel2000} for an introductory monograph.

Informally, $\logic{PDL}$ is a mixture of modal logic and algebra of regular programs. Recently, some authors have considered generalizations of modal logics to many-valued realms (see \cite{Bou2011,Caicedo2010, Fit2, Fit3, Fit1, Ost88}). These modal many-valued systems can naturally be considered as building blocks of many-valued  generalizations of $\logic{PDL}$. In other words,  these developments raise the issue of  describing the systems obtained by adding a many-valued flavor to the modal logic used to define $\logic{PDL}$. Such many-valued propositional dynamic logics would provide a language  to state correctness criteria in the form of  input/output specifications that could be \emph{partly} satisfied.

We address this problem for the modal extensions of the $n+1$-valued  \textsc{\L ukasiewicz} logics (see \cite{Luka1,Luka2,Luka3}) studied in \cite{Hansoul2006, Teheux2012}. Hence,  the truth values of the propositions  range in a set of finite cardinality $n+1$ where $n \geq 1$. 


Our starting point is the definition of a language (with test operator) for such generalizations and their corresponding $n+1$-valued \textsc{Kripke} models. In these models, relations associated to programs are crisp and valuation maps are many-valued. As an illustration of the new possibilities allowed by this language, we explain how it can be used to construct a dynamic model for formal verification of strategies  of  the \textsc{Rényi}~-~\textsc{Ulam} searching game with errors.  

The goal of this paper is the characterization of the theory of these $n+1$-valued \textsc{Kripke} models (\emph{i.e.}, the set of formulas that are true in any model). In this view, Theorem  \ref{thm:main} is our main result. It gives an axiomatization of this theory through  an $n+1$-valued  propositional dynamic deductive system that we denote by $\logic{PDL}_n$.

This result is obtained by the way of the canonical model. 
This construction defects to be a `standard'  $n+1$-valued \textsc{Kripke} model and we need a filtration result to obtain Theorem  \ref{thm:main}. 

The construction of the canonical model for $\logic{PDL}$ is algebraic in disguise. This model is built upon the set of the maximal filters of the \textsc{Lindenbaum}~-~\textsc{Tarski} algebra of $\logic{PDL}$ which is a multi-modal Boolean algebra. Naturally, the canonical model for $\logic{PDL}_n$ also has an algebraic flavor. The system $\logic{PDL}_n$ is based on modal extensions of \textsc{\L ukasiewicz} $n+1$-valued logic. Hence,  MV-algebras~-~which are the algebraic counterpart of \textsc{\L ukasiewicz} logics~-~replace Boolean algebras in this setting.

The techniques used in the proofs in this paper are generalizations of the corresponding techniques for $\logic{PDL}$. 
It is worth noting that by considering $n=1$, our results boil down to the existing ones for $\logic{PDL}$.

This paper is organized as follows. In the next section we introduce some many-valued generalizations of the language and models of $\logic{PDL}$. Section \ref{sec:ulam} provides an example that illustrates the possibilities offered by these generalizations. Section \ref{sec:systems} is devoted to the development of a sound deductive system $\logic{PDL}_n$ for the $n+1$-valued \textsc{Kripke} models.  The many-valued forms of the intrinsic axioms of $\logic{PDL}$, such as the induction axiom, are discussed when needed. Eventually, in section \ref{sec:compl} we prove the deductive completeness of $\logic{PDL}_n$ with respect to the $n+1$-valued \textsc{Kripke}-models  (proof of the filtration lemma is provided in Appendix). In order to keep the paper self-contained, we recall the necessary definitions and results about algebras of regular programs and MV-algebras.

%

\section{Many-valued \textsc{Kripke} models for dynamic logics}\label{sec:model}
The starting point of the developments of this paper is a generalization to an $n+1$-valued realm of the definitions of the propositional dynamic language and the \textsc{Kripke} models.

Let us denote by $\Pi_0$ a nonempty set of atomic programs (denoted by $a, b, \ldots$) and by $\Prop$ a countable set of propositional variables (denoted by $p, q, \ldots$). The sets $\Pi$ of programs and $\Form$ of well
formed formulas are given by the following \textsc{Backus}-\textsc{Naur} forms (where $\phi$ are formulas and $\alpha$ are programs) :
\begin{equation}
\begin{array}{c}
\phi::=p\mid0\mid  \neg \phi \mid \phi \implies \phi\mid [\alpha]\phi\\
\alpha::=a\mid\phi?\mid\alpha; \alpha\mid\alpha\cup\alpha\mid\alpha^*.
\end{array}
\end{equation}


To extend the definition of a \textsc{Kripke} model  to a $[0,1]$-valued realm, we use \textsc{\L ukasiewicz} interpretation $\implies^{[0,1]}$ and $\neg^{[0,1]}$ of the binary connector $\implies$ and the unary connector $\neg$ respectively. These maps are defined on $[0,1]$ by
\begin{equation}\label{eqn:defnluka}
\neg^{[0,1]} x =1-x \quad \mbox{ and }  \quad
x \implies^{[0,1]} y = \min(1-x+y,1).
\end{equation}

Hence,  (\ref{eqn:defnluka}) allows us to define in the obvious inductive way the \emph{$[0,1]$-interpretation} $\tau^{[0,1]}$ of any well formed formula $\tau$  constructed only with propositional variables  and connectives $\neg$ and $\implies$ (if $\tau$ has $k$ propositional variables then $\tau^{[0,1]}:[0,1]^k \rightarrow [0,1]$).
To shorten notation, when no confusion is possible we usually denote by $\neg$, $\implies$ and $\tau$ the maps $\neg^{[0,1]}$,  $\implies^{[0,1]}$ and $\tau^{[0,1]}$ respectively.

The results we are interested in are related to finitely-valued  \textsc{\L ukasiewicz} logics. It means that we only allow valuations of propositional variables in the finite subsets of $[0,1]$ that are closed for the connectors
$\neg$ and $\implies$ (and that contain $0$ and $1$). It is not difficult to
realize that these are exactly the subsets $\lucas_n=\{\frac{i}{n}\mid 0 \leq i \leq n\}$ where $n$ is an integer greater than 1 (see \cite{Cign} for details). If $\tau$ is a formula constructed from $k$ propositional variables by using only connectives $\neg$ and $\implies$, we denote by $\tau^{\mbox{\scriptsize \L}_n}$ the restriction of $\tau^{[0,1]}$ to $\lucas_n^k$.

Recall that if $R$ and $R'$ are unary relations on $W$ then the \emph{composition} $R \circ R'$ is defined by $R\circ R'=\{(u,w)\in W\times W \mid \exists v \in W (uRv \ \& \ vR'w)\}$. Moreover, the \emph{$k$-th power}  $R^k$ of $R$ is inductively defined by $R^0=\{(u, u) \mid u \in W\}$ and $R^{k+1}=R\circ R^k$ for $k \in \omega$.

\begin{defn}\label{defn:model}
  An \emph{$n+1$-valued \textsc{Kripke} model}  $\model{M}=\struc{W, R, \Val}$ is given by a nonempty
  set $W$, a map $R: \Pi_0 \rightarrow 2^{W \times W}$ that assigns a binary
  relation $R_a$ to any $a \in \Pi_0$ and a map $\Val: W
  \times \Prop \rightarrow \lucas_n$ that assigns a truth value  to any propositional variable
  $p$ of $\Prop$ in any world $w$ of $W$. 

The maps $R$ and $\Val$ are
extended by mutual induction to formulas and programs by the following
rules:
\begin{enumerate}
\item $R_{\alpha;\beta}=R_\alpha \circ R_\beta$;
\item $R_{\alpha \cup \beta}=R_\alpha \cup R_\beta$;
\item $R_{\psi?}=\{(u,u)\mid \Val(u,\psi)=1\}$;
\item\label{rule:closure}  $R_{\alpha^*}=\bigcup_{k\in \omega}(R_{\alpha})^k$;
\item $\Val(w,0)=0$;
\item \label{rule:implie}$\Val(w, \phi \implies \psi)=\Val(w, \phi) \implies^{[0,1]} \Val(w, \psi)$;
\item\label{rule:uhjn} $\Val(w, \neg \psi)=\neg^{[0,1]} \Val(w, \psi)$;
\item \label{rule:box}$\Val(w, [\alpha] \psi)=\bigwedge \{\Val(v, \psi)\mid {(w,v) \in R_\alpha}\}$.
\end{enumerate}

\end{defn}
Throughout the paper, $n$ stands for a fixed integer greater or equal to 1.  We sometimes call  \emph{\textsc{Kripke} model} an $n+1$-valued \textsc{Kripke} model.

Clearly, we intend to interpret the operator `;' as the concatenation program operator,
`$\cup$' as the alternative program operator and the operator `$*$' as the \textsc{Kleene} program operator. Hence,
if $\alpha$ and $\beta$ are programs, the connective $[\alpha]$ is read `after any execution of $\alpha$', the connective
$[\alpha \cup \beta]$ is read `after any execution of $\alpha$ or $\beta$', the connective $[\alpha;\beta]$ is read `after any 
execution of $\alpha$ followed by an execution of $\beta$' and $[\alpha^*]$ is read `after an undetermined number of executions of  $\alpha$' (rule (\ref{rule:closure}) means that $R_{\alpha^*}$ is defined as the transitive and reflexive closure of $R_\alpha$).

\begin{defn}\label{defn:true} If $w$ is a world of a \textsc{Kripke} model $\model{M}$  and if $\phi$ is a formula such that $\Val(w, \phi)=1$, we
write $\model{M}, w \models \phi$ and say that \emph{$\phi$ is true in $w$}. If
$\phi$ is a formula that is true in each world of a model $\model{M}$ then $\phi$ is
\emph{true in $\model{M}$}. A formula that is true in every \textsc{Kripke}
model is called a \emph{tautology}.
\end{defn}

We  use of the well established following abbreviations for any $\phi, \psi \in \Form$: the formula $\phi \vee \psi$ stands for $(\phi \implies \psi) \implies \psi$,  the formula $\phi \wedge \psi$
for $\neg (\neg \phi \vee \neg \psi)$, the 
formula $\phi \oplus \psi$ for $\neg \phi \implies \psi$, the formula
$\phi \odot \psi$ for $\neg (\neg \phi \oplus \neg \psi)$,  the formula $\phi \leftrightarrow \psi$  for $(\phi \implies \psi) \odot (\psi \implies \phi)$. Moreover, we assume associativity of $\oplus$ and $\odot$ (this is justified by associativity of $\oplus^{[0,1]}$ and $\odot^{[0,1]}$). Hence,  the formula $k.\psi$ and $\psi^k$ ($k \in \omega$) stands respectively  for $\psi \oplus \cdots \oplus \psi$  and $\psi \odot \cdots \odot \psi$ where the factor $\psi$ is repeated $k$ times. We adopt the convention that $\psi^0=1$ and $0.\psi=0$. It is easily checked that the resulting $[0,1]$-interpretations of these abbreviations are the following:
\begin{multicols}{2}
\begin{enumerate}
\item $x \oplus^{[0,1]} y=\min\{x+y,1\}$,
\item $x\odot^{[0,1]}y=\max\{x+y-1,0\}$,
\item $x \leftrightarrow^{[0,1]} y=1-\mid x-y \mid$,
\item $ x \vee^{[0,1]} y=\max\{x,y\}$,
\item $x \wedge^{[0,1]} y=\min\{x,y\}$.
\end{enumerate}
\end{multicols}
In \textsc{\L ukasiewicz} logic, connectors $\oplus$ and $\odot$ are respectively called \emph{strong disjunction} and  \emph{strong conjunction} because the equations $(p \odot q)^{[0,1]}\leq(p \wedge q)^{[0,1]}$ and $(p \oplus q)^{[0,1]}\geq(p \vee q)^{[0,1]}$ are satisfied. Recall that $\odot^{[0,1]}$ is a left-continuous t-norm with residuum $\neg^{[0,1]}$. This means that equation
\begin{equation}
\big((p \odot(p\implies q)) \implies q\big)^{[0,1]}=1,
\end{equation}
which can be considered as the fuzzy version of \emph{modus ponens}, is satisfied.
  It should be noted that   $(p^{k+1})^{[0,1]}\neq(p^k)^{[0,1]}$ for any $k\in \omega$ but  $(p^{k+1})^{\mbox{\scriptsize \L}_n}=(p^k)^{\mbox{\scriptsize \L}_n}$  for every $k\geq n$. Finally, the formula
$\poss{\alpha} \phi$  stands  for $\neg \necess{\alpha} \neg \phi$.


Moreover, we write $\model{M}, w \models \Gamma$ (respectively $\model{M}\models \Gamma$) if $\Gamma$ is a set of formulas that are true in $w$ (respectively in $\model{M}$).

\begin{prop}\label{prop:tauto}
The following formulas are tautologies for any programs $\alpha$ and $\beta$ (where $n$ is the integer that we have fixed to define $\lucas_n$).
\begin{multicols}{2}
 \begin{enumerate}
 \item\label{item:kjh01} $\necess{\alpha \cup \beta} p \leftrightarrow \necess{\alpha} p \wedge
     \necess{\beta} p$.\label{tauto:dgdks}
   \item\label{item:kjh02}  $\necess{\alpha ; \beta} p \leftrightarrow \necess{\alpha} \necess{\beta} p$.
   \item $\poss{\alpha \cup\beta} p \leftrightarrow \poss{\alpha} p \vee
     \poss{\beta} p$.
   \item  $\poss{\alpha ; \beta} p \leftrightarrow \poss{\alpha}
     \poss{\beta} p$.
 \item\label{item:kjh03} $\necess{q?}p \leftrightarrow (\neg q^n \vee p)$.
   \item $\necess{\alpha^*} p \implies p$.
   \item $p \implies \poss{\alpha^*} p$.
   \item $\necess{\alpha^*} p \implies \necess{\alpha} p$.
   \item $\poss{\alpha} p \implies \poss{\alpha^*} p$.
   \item\label{item:kjh03bis} $\necess{\alpha^*} p \leftrightarrow (p \wedge \necess{\alpha}
   \necess{\alpha^*} p)$.
   \item$\poss{\alpha^*} p \leftrightarrow (p \vee \poss{\alpha}
   \poss{\alpha^*} p)$.
   \item\label{item:kjh04} $(p \wedge \necess{\alpha^*} (p \implies \necess{\alpha}
     p)^n) \implies \necess{\alpha^*} p$.
\item\label{item:kjh05} $[\alpha^*]p\implies[\alpha^*][\alpha^*]p$.
\newcounter{temp}
\setcounter {temp}{\value{enumi}}
\end{enumerate}
\end{multicols}

Moreover, the following formulas are tautologies for any program $\alpha$, because they are tautologies
of the modal $n+1$-valued \textsc{\L ukasiewicz} logic.
\begin{enumerate}
\setcounter{enumi}{\value{temp}}
\item\label{item:K} $\necess{\alpha} (p\implies q) \implies (\necess{\alpha} p \implies
  \necess{\alpha} q)$.
 \item $\necess{\alpha} (p \wedge q)\leftrightarrow\necess{\alpha} p \wedge \necess{\alpha}
   q$ and $\poss{\alpha} (p \vee q)\leftrightarrow\poss{\alpha} p \vee \poss{\alpha} q$.
\item $(\necess{\alpha}p \vee \necess{\alpha}q) \implies \necess{\alpha}(p \vee q)$.
\item $([\alpha] \phi \odot \poss{\alpha} \psi) \implies \poss{\alpha}(\phi \odot \psi)$.
\item\label{tem:tauto} $\poss{\alpha}(\phi \odot \psi)\implies (\poss{\alpha}\phi \odot \poss{\alpha} \psi)$
\item\label{item:tau_respect} If $\tau(q)$ is a formula with a single variable $q$  which is constructed only with the connectors
  $\neg$ and $\implies$ and whose  $[0,1]$-interpretation is increasing then $\tau(\necess{\alpha} p) \leftrightarrow
  \necess{\alpha} \tau(p)$ and $\tau(\poss{\alpha} p) \leftrightarrow
  \poss{\alpha} \tau(p)$.
 \end{enumerate}
\begin{ex}\label{ex:nonindu}
It is worth noting that the formula $\big(p \wedge \necess{\alpha^*}(p\implies \necess{\alpha}p)\big)\implies \necess{\alpha^*}p$ is not a tautology. It would have been the most natural many-valued generalization of  the Induction Axiom of $\logic{PDL}$. As a counterexample, consider the model $\model{M}=\struc{\{u,v\}, R, \Val}$ where $R_a=\{(u,v)\}$, $\Val(u,p)=3/4$ and $\Val(v,p)=1/4$. It follows that on the one hand $\Val(u,\necess{a^*}p)=\Val(u,p)\wedge \Val(v,p)=1/4$. On the other hand, we obtain successively
\begin{eqnarray}
\Val(u, \necess{\alpha^*}(p\implies \necess{\alpha}p))& = &  \Val(u, p \implies \necess{a}p)\wedge\Val(v, p \implies \necess{a}p)\\
& = & 1/2 \wedge  1\\
& = & 1/2.
\end{eqnarray}
It follows that $\Val(u, p \wedge \necess{\alpha^*}(p\implies \necess{\alpha}p))=3/4\wedge 1/2=1/2 \not=1/4=\Val(u,\necess{a^*}p)$.
\end{ex}
\end{prop}

\section{An illustration, the \textsc{R\' enyi} - \textsc{Ulam} game }\label{sec:ulam}

We can use the previously defined models to provide a framework for an
analysis  of the famous \textsc{Rényi}~-~\textsc{Ulam}
game. \textsc{Ulam}'s formulation of the game in \cite{Ulam}, which was previously and
independently introduced by \textsc{Rényi}, is the following:
\begin{quotation}
Someone thinks of a number between one and one million (which is just less
than $2^{20}$). Another person is allowed to ask up to twenty questions, to
each of which the first person is supposed to answer only yes or no. Obviously
the number can be guessed by asking first: is the number in the first
half-million? and again reduce the reservoir of numbers in the next question
by one-half, and so-on. Finally, the number is obtained in less than $\log_2 1
000 000$. Now, suppose that one were allowed to lie once or twice, then how
many questions would one need to get the right answer? 
\end{quotation}

Many researchers (mainly computer scientists) have  focused their attention on
that game since the publication of \textsc{Ulam}'s book \cite{Ulam}.  The success 
of the game is due to its connections with the theory of error-correcting 
codes with feedbacks in a noisy channel and the complexity of the problem of
defining optimal strategies for the game. We refer to \cite{Pelc} for an overview
of the literature about the \textsc{Rényi}~-~\textsc{Ulam} game.

The game has also been considered by many-valued logicians as a way to give a
concrete interpretation of \textsc{\L ukasiewicz} finitely-valued calculi and
their associated algebras (see \cite{Mundici1992}). Mathematicians have modeled
the game by coding algebraically \emph{questions} and
\emph{answers}. We recall this model, which is due to
\textsc{Mundici}, and then build a dynamic layer upon it in order  to model the interactions
between the two gamers.

\subsection{Algebraic approach of the states of knowledge}
We call the first gamer (the one who chooses a number and can lie) Pinocchio,
and the second gamer Geppetto. Let us denote by $M$ the \emph{search space}\index{search space},
\emph{i.e.}, the finite set of integers (or whatever) in which Pinocchio can pick up
his number. Let us also assume that Pinocchio can lie $n-1$ times. 

We set up a way to algebraically encode the information
defined by Pinocchio's answers, \emph{i.e.}, to model Geppetto's state of knowledge of
the game after each of Pinocchio's answers. This can be done by considering
at step $i$ of the game (after $i$ answers) the map
$r_i: M \rightarrow \{0,1, \ldots, n\}
$
where $r_i(m)$ is the number of the $i$ previous answers that refute the element $m$
of $M$ as Pinocchio's number. Indeed, once $r(m)=n$, since Pinocchio is
allowed to lie $n-1$ times, Geppetto can safely conclude that $m$ is not the
`right' number. Hence, the game ends once Geppetto encodes its knowledge by
a map $r$ which is equal to $n$ in any element $m$ of $M$ but in the searched number. 


In order to introduce \textsc{\L ukasiewicz} language in the interpretation of the game, we 
consider an equivalent representation of Geppetto's states of knowledge. This approach  was introduced in \cite{Mundici1992}.
\begin{defn}
 A \emph{state of knowledge}\index{state of knowledge} is a map $f:M \rightarrow \lucas_n$. The \emph{state of
 knowledge $f$ at some step of the game} is defined
 by $f(m)=1-\frac{r(m)}{n}$ where $r(m)$ denotes for any $m$ in $M$ the number of Pinocchio's
 answers that refute $m$ as the searched number.
\end{defn}

Hence, informally speaking, if $f$ is a state of knowledge at some step of the game, the number
$f(m)$ can be viewed for any $m$ in $M$ as the relative distance between $m$ and the set of the
elements of $M$ that can be safely discarded as inappropriate.
\subsection{Questions and answers}
 Note that during the game any question  is equivalent to a
question of the form `Does the searched number belong to $Q$?' for a subset
$Q$ of the search space $M$. Hence, for the remainder of  this section, we denote any
question by its associated subset $Q$ of $M$.

Let us assume that Geppetto has reached the state of knowledge $f$ and that he
asks question $Q$.  What is the state of knowledge $f'$ of the game after
Pinocchio's answer? If Pinocchio answers positively (`Yes, the number
belongs to $Q$') then Gepetto increments $r(m)$ by one (if necessary) for any $m$ in $M \setminus Q$  since a positive answer to $Q$ is
equivalent to a negative answer to $M\setminus Q$, \emph{i.e.},
\begin{equation}
f':M \rightarrow \lucas_n= m \mapsto \left\{
  \begin{array}{ll}
    f(m) & \mbox{if } m \in Q\\
    \max \{f(m)-\frac{1}{m}, 0\} & \mbox{if } m \in M \setminus Q.  \end{array}\right.
\end{equation}
On the contrary, if Pinocchio answers negatively to $Q$, then  Gepetto
increments $r(m)$ by one (if necessary) for any $m$ in $Q$, \emph{i.e.}, 
\begin{equation}
f':M \rightarrow \lucas_n= m \mapsto \left\{
  \begin{array}{ll}
    f(m) & \mbox{if } m \in M \setminus Q\\
    \max \{f(m)-\frac{1}{m},0\} & \mbox{if } m \in Q.  \end{array}\right.
\end{equation}
This line of argument justifies the following definition.

\begin{defn}
  If $Q$ is a subset of $M$, the \emph{positive answer} to $Q$ is the map
\[
f_Q:M \rightarrow \{\frac{n-1}{n}, 1\}:m \mapsto \left\{
  \begin{array}{ll}
    1 & \mbox{if } m \in Q \\ 
   \frac{n-1}{n} & \mbox{if } m \in M\setminus Q.
  \end{array}
\right.
\]
The \emph{negative answer} to $Q$ is the positive answer $f_{M\setminus Q}$ to
$M\setminus Q$.
\end{defn}

We can thus encode algebraically any of Pinocchio's answers. Recall that the interpretation of the binary connector $\odot$ on $[0,1]$ is defined by $x \odot^{[0,1]} y=\max(x+y-1,0)$.
\begin{fact}
  Assume that Geppetto has reached the state of knowledge $f$ and that he asks
  question $Q$. After Pinocchio's answer to $Q$, the stage of knowledge $f'$ of the
  game is $f\odot f_Q$ if Pinocchio's answer is positive and $f \odot
  f_{M\setminus Q}$ if it is negative.
\end{fact}
\subsection{A dynamic layer}
Roughly speaking, we have modeled the game in a static way. There
is no structure  to model the possible \emph{sequences} of states of games. We provide such a structure through the Question/Answer relations on the set of the states of knowledge. 
The  atomic programs are the possible questions, \emph{i.e.} $\Pi_0=2^M$. The set
of propositional variables $\{p_m \mid m \in M\}$ that are relevant to the
problem is made of a variable $p_m$ for any $m$ in $M$ that can be read as
`$m$ is far from the set of rejected element'
or `the relative distance between $m$ and the set of rejected elements is'.

\begin{defn}
  The model of the \textsc{Rényi}~-~\textsc{Ulam} game with search space $M$
  and $n-1$ lies is the $n+1$-valued \textsc{Kripke} model 
  $\model{M}=\struc{\lucas_n^M, R, \Val}$
where
\begin{enumerate}
\item for any $Q$ in $2^M$, the relation $R_Q$ contains $(f,f')$ if $f'=f\odot
  f_Q$ or if $f'=f\odot f_{M\setminus Q}$,
\item for any $m$ in $M$ and any $f$ in $\lucas_n^M$, we set $\Val(f,p_m)=f(m)$.
\end{enumerate}
\end{defn}

This model provides a way to interpret any run of the game as a path from the
initial state $f:m \mapsto 1$ to any final winning state.

\begin{ex}
Examples of formulas that state correctness specifications for `honest' sequences of states of knowledge
include the following. We denote by $\tau_{i/n}(p)$  a formula whose interpretation on $\lucas_n$ is valued in $\{0,1\}$ and satifies $\tau_{i/n}^{{\tiny\lucas}_n}(x)=1 \iff x\geq i/n$. See Definition \ref{defn:ghu} for a formal definition. 
\begin{enumerate}
\item $[Q]p_m \implies p_m$
\item $\tau_{\frac{i}{n}}(p_m) \implies [Q;M\setminus Q] \tau_{\frac{i-2}{n}}(p_m)$ (if we agree that $\tau_{\frac{i-2}{n}}(p_m)=1$ if $i-2 \leq 0$).
\end{enumerate}
\end{ex}

As mentioned in the introduction, it is not the purpose of this paper to push further the investigation of the new possibilities allowed by the $n+1$-valued \textsc{Kripke} models. Nevertheless, we give some ideas of possible applications in section \ref{sec:concl}.

\section{$n+1$-valued propositional dynamic logics}\label{sec:systems}

We aim to provide a set of rules that allow to syntacticly   generate the theory of the $n+1$-valued \textsc{Kripke} models defined in section \ref{sec:model}. The underlying modal system on which we base the following definition is the modal $\lucas_n$-valued logic introduced in \cite{Hansoul2006,Teheux2012}.

\begin{defn}\label{defn:logic}
  An \emph{$n+1$-valued propositional dynamic logic} (or simply a
  \emph{logic}) is a subset $\logic{L}$ of $\Form$ that is
closed under the rules of \emph{modus ponens}, uniform substitution and  necessitation (generalization) and
that contains the following axioms:
\begin{enumerate}
\item tautologies of the $n+1$-valued \textsc{\L ukasiewicz} logic;
\item\label{item:defnmodal} for any program $\alpha$,  axioms defining modality $[\alpha]$: 
\vspace{-1em}
\begin{multicols}{2}
\begin{enumerate}
\item $\necess{\alpha} (p\implies q) \implies (\necess{\alpha} p \implies
  \necess{\alpha} q)$,
\item $\necess{\alpha} (p \oplus p) \leftrightarrow
  \necess{\alpha} p \oplus \necess{\alpha} p$,
\item $\necess{\alpha} (p \odot p) \leftrightarrow
  \necess{\alpha} p \odot \necess{\alpha} p$,
\end{enumerate}
\end{multicols}
\vspace{-1em}
\item\label{item:defnprogram} the axioms that define the program operations: for any
  programs $\alpha$ and $\beta$;
\vspace{-1em}
\begin{multicols}{2}
\begin{enumerate}
\item  $\necess{\alpha \cup
    \beta} p \leftrightarrow \necess{\alpha} p \wedge \necess{\beta} p$,
\item $\necess{\alpha ; \beta} p \leftrightarrow \necess{\alpha} \necess{\beta} p$,
\item $[q?] p \leftrightarrow (\neg q^n \vee p )$,
\item $\necess{\alpha^*}p
  \leftrightarrow (p \wedge \necess{\alpha}  \necess{\alpha^*} p)$,
\item $[\alpha^*]p \implies [\alpha^*][\alpha^*]p$,
\end{enumerate}
\end{multicols}
\vspace{-1em}
\item\label{item:defnindu} the induction axiom $\big(p \wedge \necess{\alpha^*}(p \implies
  \necess{\alpha}p)^n\big)\implies \necess{\alpha^*}p$ for any program $\alpha$.
\end{enumerate}
We denote by $\logic{PDL}_n$ the smallest  $n+1$-valued propositional dynamic logic.

As usual, a formula $\phi$ that belongs to a logic $\logic{L}$ is called a
\emph{theorem} of $\logic{L}$ and we often write $\proves \phi$ instead of $\phi \in \logic{PDL}_n$. 
\end{defn}

Note that formulas of item (\ref{item:defnmodal}) of Definition \ref{defn:logic} are tautologies according to items (\ref{item:K}) and (\ref{item:tau_respect}) of Proposition \ref{prop:tauto}. Similarily, formulas in (\ref{item:defnprogram}) and (\ref{item:defnindu}) of Definition \ref{defn:logic} are formulas (\ref{item:kjh01}), (\ref{item:kjh02}), (\ref{item:kjh03}), (\ref{item:kjh03bis}), (\ref{item:kjh04}), (\ref{item:kjh05}) of Proposition \ref{prop:tauto}.  

\begin{rem}\label{rem:tauto_modal}
Note that conditions (1) and (2) and the deduction rules of  Definition \ref{defn:logic} together with deductive completeness for the modal $\lucas_n$-valued logic (see Theorem 6.2 in \cite{Teheux2012}) ensure that if $\psi$ is  a tautology of the modal $\lucas_n$-valued \textsc{\L ukasiewicz} logic and if $\alpha \in \Pi$ then the formula obtained from $\psi$ by substitution of any occurrence of   $\square$ by $[\alpha]$ is a theorem of $ \logic{PDL}_n$.
\end{rem}

Informally, the induction axiom (4) means `if after an undetermined number of executions of $\alpha$ the truth value of $p$ cannot
decrease after a new execution of $\alpha$, then the truth value of $p$ cannot
decrease after any undetermined number of executions of $\alpha$'. Hence, it is a natural
generalization of the induction axiom of $\logic{PDL}$ (which
could not have been adopted without modification according to Example \ref{ex:nonindu}).

Let us introduce some notations in order to comment the axioms $[\alpha](p\oplus p)\leftrightarrow ([\alpha] p \oplus [\alpha] p)$ and $[\alpha](p\odot p)\leftrightarrow ([\alpha] p \odot [\alpha] p)$.
\begin{defn}\label{defn:ghu}
Let $i$ be an element of $\{1, \ldots, n\}$. We denote by $\tau_{i/n}$ a composition (fixed throughout the paper) of the formulas $p\oplus p$ and $p \odot p$ whose interpretation on $\lucas_n$ is defined by $\tau_{i/n}^{{\tiny\lucas}_n}(x)=0$ if $x<\frac{i}{n}$ and $\tau_{i/n}^{{\tiny\lucas}_n}(x)=1$ if $x\geq \frac{i}{n}$ (see \cite{Ost88} for the existence and the construction of such formulas). 

For any $i \in \{0, \ldots, n\}$, we denote by $I_{i/n}$ the formula $\tau_{i/n}\wedge\neg \tau_{(i+1)/{n}}$ (where we set $\tau_{(n+1)/n}=\tau_{0/n}=p\oplus\neg p$).
\end{defn}

Hence, the interpretation on $\lucas_n$ of $I_{i/n}$  is the characteristic function of $\{\frac{i}{n}\}$. The following result is a consequence of deductive completeness for  modal $\lucas_n$-valued \textsc{\L ukasiewicz} logic (see \cite{Teheux2012}).

\begin{fact}
In the definition of $\logic{PDL}_n$, for any $\alpha \in \Pi$, the pair of axioms 
\begin{equation}
\{[\alpha](p\star p)\leftrightarrow ([\alpha] p \star [\alpha] p)\mid \star\in\{\odot, \oplus\}\}
\end{equation}
 can be equivalently replaced by the axioms  
\begin{equation}
\{[\alpha]\tau_{i/n}(p)\leftrightarrow\tau_{i/n}([\alpha]p)\mid i \in\{1, \ldots, n\}\}.
\end{equation}
\end{fact}

Hence, informally speaking, the content of the pair of  axioms  $\{[\alpha](p\star p)\leftrightarrow ([\alpha] p \star [\alpha] p)\mid \star\in\{\odot, \oplus\}\}$ is essentially the following.
\begin{quotation}
 For any $i\leq n$,  the truth value of the statement `after any execution of $\alpha$, formula $\phi$ holds' is at least $\frac{i}{n}$ if and only if it holds that `after any execution of $\alpha$ the truth value of $\phi$ is at least $\frac{i}{n}$' .
\end{quotation}

Proposition \ref{prop:tauto} states that the axioms of $\logic{PDL}_n$ are tautologies.
Tautologies are preserved by application of the deduction rules. It follows that any theorem of $\logic{PDL}_n$ is a tautology.

As an illustration of Definition \ref{defn:logic}, we prove that $\logic{PDL}_n$ is closed under a loop invariance rule. We say that a rule of inference is \emph{derivable} in $\logic{PDL}_n$ if its consequence can be obtained from its premises by application of rules and axiom schemes of $\logic{PDL}_n$.

\begin{lem}\label{lem:IND}
For any $\alpha \in \Pi$, the  rule
\[
\mathrm{(LI)}\quad  \inferrule{(\phi \implies [\alpha]\phi)^n}{(\phi \implies [\alpha^*]\phi)}
\]
 is derivable in $\logic{PDL}_n$.
\end{lem}
\begin{proof}
Assume that $\proves (\phi \implies [\alpha] \phi)^n$.
Then 
\begin{eqnarray}
& \proves & [\alpha^*](\phi \implies [\alpha] \phi)^n\label{eqn:jsf}\\\
& \proves & [\alpha^*](\phi \implies [\alpha] \phi)^n\implies\big(\phi \implies (\phi \wedge [\alpha^*](\phi \implies [\alpha]\phi)^n)\big)\label{eqn:qdd}
\end{eqnarray}
where (\ref{eqn:jsf}) is obtained by generalization and \ref{eqn:qdd} by the fact that $p \implies (t \implies (p \wedge t))$ is a tautology of the $n+1$-valued \textsc{\L ukasiewicz} logic (and we apply substitution $p:=[\alpha^*](\phi \implies [\alpha] \phi)^n$ and $t:=\phi$). It follows that
\begin{eqnarray}
& \proves &  \phi \implies (\phi \wedge [\alpha^*](\phi \implies [\alpha]\phi)^n)\label{eqn:qdd01}\\
& \proves &  \phi \implies \necess{\alpha^*}\phi\label{eqn:qdd02}
\end{eqnarray}
where (\ref{eqn:qdd01}) is obtained by \emph{modus ponens} and (\ref{eqn:qdd02}) by double \emph{modus ponens} and induction axiom applied to the tautology of the $n+1$-valued \textsc{\L ukasiewicz} logic $(p\implies q)\implies ((q \implies t) \implies (p\implies t))$ with substitution $p:=\phi$,  $q:=\phi \wedge [\alpha^*](\phi \implies [\alpha]\phi)^n$ and $t:=\necess{\alpha^*}\phi$.
\end{proof}
\begin{rem}
We say that a rule of inference $\mathrm{RI}$ is \emph{admissible} in $\logic{PDL}_n$ if the system formed by $\logic{PDL}_n$ and $\mathrm{RI}$ has the same theorems as $\logic{PDL}_n$. Since for any $k \in \omega$ the rule $\phi/\phi^k$ is admissible in \textsc{\L ukasiewicz} $n+1$-valued logic, we can deduce from Lemma \ref{lem:IND} that the rule
\[
\mathrm{(LI^\sharp)}\quad  \inferrule{(\phi \implies [\alpha]\phi)}{(\phi \implies [\alpha^*]\phi)}
\]
is admissible in $\logic{PDL}_n$.
\end{rem}


\section{Deductive Completeness for $\logic{PDL}_n$}\label{sec:compl}
The main result of the paper is Theorem \ref{thm:main} that states that $\logic{PDL}_n$ is complete with respect to the $n+1$-valued \textsc{Kripke} models. To obtain this result, we use the technique of the canonical model. We follow Part II of \cite{Harel2000} to guide us in our constructions and developments.

As in the case of propositional dynamic logic, in the construction of the canonical model for $\logic{PDL}_n$ , the relation associated to a program is not built inductively from the relations associated to its atomic programs. Instead, we directly associate to each $\alpha$ of $\Pi$ a  relation $R_\alpha$ defined in a canonical way. In fact, the inductive rules involving the operators `$;$', `$\cup$' and `$?$' are satisfied in the canonical model, but $R_{\alpha^*}$ may strictly contain the transitive and reflexive closure of $R_\alpha$. We use the technique of filtration to construct $\lucas_n$-valued \textsc{Kripke} models from this canonical model.

\subsection{Filtration lemma}
The canonical model of $\logic{PDL}_n$ will turn out to be non standard in the following sense.
\begin{defn}\label{defn:nonstand}
  A \emph{weak non standard $n+1$-valued \textsc{Kripke} model}
  $\model{M}=\struc{W, R, \Val}$ is given by a nonempty set $W$ a map $R: \Pi
  \rightarrow 2^{W \times W}$ and a valuation map $\Val: W \times \Prop \rightarrow
  \lucas_n$. The valuation map is extended to formulas by way of the rules (\ref{rule:implie}), (\ref{rule:uhjn}) and (\ref{rule:box}) of Definition \ref{defn:model}. If $w \in W$ and $\phi \in \Form$, we write $\model{M}, w \models \phi$ if $\Val(w, \phi)=1$. We write $\model{M}\models \phi$ if $\model{M}, w \models \phi$ for any $w$ in $W$. 

A \emph{non standard $n+1$-valued \textsc{Kripke} model} is a {weak non standard $n+1$-valued \textsc{Kripke} model}
$\model{M}=\struc{W, R, \Val}$ such that for any programs $\alpha$ and $\beta$ and any formula
  $\psi$,
\begin{enumerate}
\item\label{cdt:Rnonstand00} the following identities are satisfied in $\model{M}$:
\begin{enumerate}
\item $R_{\alpha;\beta}=R_\alpha
  \circ R_\beta$,
\item $R_{\alpha \cup \beta}=R_\alpha \cup R_\beta$,
\item $R_{\psi ?}=\{(u,u)\mid\Val(u, \psi)=1\}$;
\end{enumerate}
\item\label{cdt:Rnonstand} the relation $R_{\alpha}^*$ is a
  transitive and reflexive extension of $R_\alpha$;
\item\label{cdt:Rnonstand01} For any $\phi\in \Form$,  $\model{M}\models\{ [\alpha^*]\phi\implies (\phi\wedge [\alpha][\alpha^*]\phi), [\alpha^*]\phi\implies [\alpha^*][\alpha^*]\phi, (\phi \wedge \necess{\alpha^*} (\phi \implies \necess{\alpha}
     \phi)^n) \implies \necess{\alpha^*} \phi \}$.
\end{enumerate}
\end{defn}

Note that in condition (\ref{cdt:Rnonstand}) of the previous definition we allow $R_{\alpha^*}$ to be \emph{any} reflexive and transitive extension of $R_\alpha$.
 
\begin{rem}\label{rem:sound}
Conditions (\ref{cdt:Rnonstand00}) and (\ref{cdt:Rnonstand01}) ensure that  if $\phi$ is a theorem of $\logic{PDL}_n$ then $\model{M} \models \phi$ for any $n+1$-valued non standard \textsc{Kripke} model $\model{M}$ (because axioms and rules of $\logic{PDL}_n$ are sound for non standard \textsc{Kripke} frames).
\end{rem}

Filtration lemmas are usually proved by induction on the subformula relation. In ($n+1$-valued) propositional dynamic logic, the use of induction is  somehow cumbersome because of the interdependence of the definitions of formulas and programs.  We use the \textsc{Fischer~-~Ladner} closure $\fish{\phi}$ of a formula $\phi$
to prove a filtration lemma for $n+1$-valued non standard models. 

To ease readability, proof of the Filtration Lemma is moved in Appendix in which we also recall the definition (Definition \ref{defn:closure}) of the \textsc{Fisher~-~Ladner} closure of a formula (see also \cite{Harel2000}) . 
%

\begin{defn}
  If $\model{M}=\struc{W, R, \Val}$ is a weak $n+1$-valued non standard \textsc{Kripke} model and
  if $\phi$ is a formula then we define the equivalence relation $\equiv_\phi$
  on $W$ by
\begin{equation}
u \equiv_\phi v \quad \mbox{ if } \quad \forall  \psi \in \fish{\phi} \ \Val(u, \psi)=\Val(v,\psi).
\end{equation}

We denote by $[W]_\phi$ (or simply by $[W]$) the quotient of $W$ by
$\equiv_\phi$ and by $[u]_\phi$ (or simply $[u]$) the class of an element $u$
of $W$ for $\equiv_{\phi}$.

Then, for any atomic program $a$ of $\Pi_0$ we define the relation $R_a^{[\model{M}]_\phi}$  by 
\begin{equation}\label{eqn:filtre01}
R_a^{[\model{M}]_\phi}=\{([u], [v]) \mid (u,v)\in R_a\}\end{equation}
 and the valuation  map $\Val^{[W]}$ on $[W] \times \Prop$ by
\begin{equation}\label{eqn:filtre02}
\Val^{[\model{M}]_\phi}([u],p)=\bigvee \Val([u],p).
\end{equation}

The $n+1$-valued \textsc{Kripke} model $[\model{M}]_\phi=\struc{[W]_\phi, R^{[\model{M}]_\phi}, \Val^{[\model{M}]_\phi}}$ is called \emph{the  filtration of $\model{M}$ through $\phi$.}  If no confusion is possible we prefer to denote this model by $[\model{M}]=\struc{[W], R^{[\model{M}]}, \Val^{[\model{M}]}}$.
\end{defn}
Note that the number of worlds in $[\model{M}]_\phi$ is finite and bounded by $(n+1)^{\card{\fish{\phi}}}$.

The proof of the following result is provided in Appendix \ref{sect:appendix}.

\begin{lem}[Filtration]\label{lem:filtration} Assume that $\model{M}=\struc{W, R, \Val}$ is  an $n+1$-valued non
  standard \textsc{Kripke} model and that $\phi$ is a formula.
  \begin{enumerate}
  \item If $\psi$ is in $\fish{\phi}$ then $\Val(u,\psi)=\Val^{[\model{M}]}([u], \psi)$.
  \item For every $\necess{\alpha} \psi$ in $\fish{\phi}$,
    \begin{enumerate}
    \item if $(u,v) \in R_\alpha$ then $([u],[v]) \in R_\alpha^{[\model{M}]}$;
    \item  if $([u],[v]) \in
    R_\alpha^{[\model{M}]}$ then  $\Val(u, \necess{\alpha} \psi)  \leq \Val(v, \psi)$.
    \end{enumerate}
  \end{enumerate}
\end{lem}

We obtain the decidability of the satisfiability problem for $\logic{PDL}_n$ as an immediate consequence of Lemma \ref{lem:filtration}.

\begin{defn}\label{defn:satis}
A formula $\phi$ of $\Form$ is \emph{satisfiable} if there is an $n+1$-valued \textsc{Kripke} model and a world in this model in which $\phi$ is true.
\end{defn}

\begin{cor} The problem of deciding if a formula of $\Form$ is satisfiable is decidable.
\end{cor}
\begin{proof}
If $\phi$ is satisfiable in an $n+1$-valued \textsc{Kripke} model, Lemma  \ref{lem:filtration} ensures that it is satisfiable in a model with at most $(n+1)^{\card{\fish{\phi}}}$ worlds.
\end{proof}




\subsection{The canonical model}
We construct the canonical $n+1$-valued  \textsc{Kri\-pke} model of
$\logic{PDL}_n$ on the set of homomorphisms from the
\textsc{Lindenbaum~-~Tarski} algebra of $\logic{PDL}_n$ to $\lucas_n$. We assume that the
reader has some acquaintance with the theory of MV-algebras which are the
algebras of the many-valued \textsc{\L ukasiewicz} logics. We only recall the necessary definitions. See \cite{Gispert2005} for an introduction or \cite{Cign} for a monograph on the subject. 

Recall that the variety $\var{MV}$ of MV-algebras is generated by the algebra $\struc{[0,1], \implies, \neg, 1}$ where $\neg$ and $\implies$ are defined on $[0,1]$ as their \textsc{\L ukasiewicz} interpretation (see section \ref{sec:model}). $\var{MV}$ can be described as the class of algebras $A=\struc{A,\implies, \neg,1}$ of type $(2,1,0)$ that satisfy the following equations\footnote{This axiomatization is not the most commonly used axiomatization of $\var{MV}$, but it is the most efficient for our purpose. See \cite{Cign} for details.}:
\begin{equation}\label{eqn:axiomMV}
\begin{array}{ll}
x\implies 1=x, & (x\implies y)\implies ((y\implies z)\implies (x\implies z))\!=\!1,\\
(x\implies)\implies y = (y \implies x) \implies x, & (\neg x \implies \neg y)\implies (y \implies x) =1. 
\end{array}
\end{equation}

 The variety $\var{MV}_n$ is the subvariety of $\var{MV}$ generated by the subalgebra $\lucas_n$ of $[0,1]$. We denote by $\var{MV}(A, \lucas_n)$ the set of the MV-algebra homomorphisms from $A$ to $\lucas_n$ for any $A\in \var{MV}_n$. 

In any MV-algebra $A$, the relation $\leq$ defined by $a\leq b$ if  $a\implies b=1$ is a bounded distributive lattice order on $A$. The variety $\var{MV}$ was introduced by \textsc{Chang} (see \cite{Chang1, Chang2}) in order to obtain an algebraic completeness result for \textsc{\L ukasiewicz} infinite-valued logic.

\begin{defn}
We denote by $\free_n$ the \textsc{Lindenbaum~-~Tarski} algebra of $\logic{PDL}_n$, that is, the quotient of $\Form$ by the syntactic equivalence relation $\equiv$ defined by $\phi \equiv \psi$ if $\logic{PDL}_n\proves \phi \leftrightarrow \psi$. This quotient is equipped with  the operations $\implies$, $\neg$ and $[\alpha]$ ($\alpha \in \Pi$) defined in the obvious way: $(\phi/ \equiv) \implies (\psi/ \equiv)=(\phi \implies \psi)/\equiv$, $\neg (\psi/\equiv)=(\neg \psi)/\equiv$ and $[\alpha](\psi/\equiv)=([\alpha]\psi)/\equiv$.
\end{defn}

For the sake of readability, we prefer to denote by $\phi$ the class $\phi/\equiv$.

\begin{lem}\label{lem:IamFree}
 The reduct of $\free_n$ to the language $\{\implies, \neg, 1\}$ belongs to $\var{MV}_n$. 
\end{lem}
\begin{proof}
We have included the tautologies of  \textsc{\L ukasiewicz} $n+1$-valued logic in our axiomatization of $\logic{PDL}_n$.
\end{proof}

The preceding lemma leads to the definition of the canonical model for $\logic{PDL}_n$. The classical construction of the canonical model for $\logic{PDL}$ is based on the set of the maximal (Boolean) filters of the \textsc{Lindenbaum~-~Tarski} algebra $\free$ of $\logic{PDL}$. One of the  key element of this construction is a separation result (a consequence of the Ultrafilter Theorem) that states that  for any $\phi, \psi$ such that $\phi, \psi, \phi\leftrightarrow\psi \not\in \logic{PDL}$, there is a maximal filter of $\free$ that contains $\phi/\equiv$ but not $\psi/\equiv$. We can state this result using homomorphisms. Indeed, a subset $F$ of a Boolean algebra $A$ is a (proper) maximal filter if and only if the map $\pi_F:A\rightarrow \bold{2}$ (where $\bold{2}$ denotes the two element Boolean algebra) defined by $\pi_F^{-1}(1)=F$ is an homomorphism. Hence, the separation result can be stated in this way: for any $\phi, \psi$ such that $\phi,\psi, \phi\leftrightarrow\psi \not\in \logic{PDL}$, there is an homomorphism $v:\free\rightarrow \bold{2}$ such that $v(\phi)=1$ and $v(\psi)=0$.

There is an analogous separation result for the variety $\var{MV}_n$: if $A\in \var{MV}_n$ and $a\neq b \in A$, there is an homomorphism $v:A\rightarrow \lucas_n$ such that\footnote{Without going into details, note that this is a consequence of the characterization of subdirectly irreducible elements of $\var{MV}_n$. See \cite{Cign}.} $v(a)\neq v(b)$. This result, together with Lemma \ref{lem:IamFree},  indicates that the set $\var{MV}(\free_n, \lucas_n)$ is a good candidate for the universe of the canonical model of $\logic{PDL}_n$. Before proceeding with the construction of this model, let us recall how  $\var{MV}(A, \lucas_n)$ is linked with the set of maximal filters of $A\in \var{MV}_n$.

A \emph{filter}  of an MV-algebra $A$ is a subset $F$ of $A$  that
contains $1$ and that contains $y$ whenever it contains $x$ and $x \implies
y$. Equivalently, a filter of $A$ is a nonempty increasing subset of $A$ closed under $\odot$. If $X$ is a nonempty subset of an MV-algebra $A$, the filter generated by $X$ is the filter
\begin{equation}
 \struc{X}=\{b \in A \mid \exists k \in \omega, \epsilon \in \omega^k, x \in X^k (b \geq x_1^{\epsilon_1} \odot \cdots \odot x_k^{\epsilon_k}) \}.
\end{equation}
 Filters are ordered by set inclusion and the proper maximal elements are called \emph{maximal filters} and correspond to homomorphisms from $A$ to $\lucas_n$ in the following way. For any maximal filter $F$ of $A\in \var{MV}_n$, there is only one homomorphism  $v_F:A\rightarrow \lucas_n$ that satisfies $v_F^{-1}(1)=F$.
The map $v_{\cdot}:F\mapsto v_F$ has converse $\cdot^{-1}(1):v\mapsto v^{-1}(1)$ that associates a maximal filter for any $v\in\var{A}(A, \lucas_n)$.

\begin{defn}\label{defn:canomod}
  The \emph{canonical model} of $\logic{PDL}_n$ is defined as  the model $\model{M}^{c}=\struc{W^{c}, R^{c}, \Val^{c}}$
where
\begin{enumerate}
\item  $W^{c}= \var{MV}(\free_n,
  \lucas_n)$;
\item if $\alpha \in \Pi$, the relation $R^{c}_\alpha$ is defined as 
\[R^{c}_\alpha=\{(u,v)\mid \forall \phi \in  \free_n      \ \big(u(\necess{\alpha} \phi)=1 \Rightarrow v(\phi)=1\big)\};\]
\item the map $\Val^{c}$ is defined as \[\Val^{c}:W^c \times \Form: (u,\phi)\mapsto u(\phi).\]
\end{enumerate}
When no confusion arises, we prefer to write $W$, $R$ and $\Val$ instead of
$W^{c}$, $R^{c}$ and $\Val^{c}$ respectively.
\end{defn}
\begin{lem}\label{lem:alterative_def}
If $\alpha \in \Pi$, then
\begin{equation}
R_\alpha^c=\{(u,v) \mid \forall \phi \in \free_n (v(\phi)=1\Rightarrow u(\poss{\alpha}\phi)=1)\}.
\end{equation}
\end{lem}
\begin{proof}
Assume that $(u,v)\in R_\alpha^{c}$ and that $\phi$ is an element of $\free_n$ such that $v(\phi)=1$. If $u(\poss{\alpha}\phi)<1$  then $u(\necess{\alpha}\neg \phi)=1-u(\poss{\alpha}\phi)>0$.  Let $i$ be the element of $\{0, \ldots, n-1\}$ such that $u(\necess{\alpha}\neg \phi)=\frac{i}{n}$. It follows that $\tau_{i/n}(u([\alpha]\neg \phi))=u([\alpha]\tau_{i/n}(\neg \phi))=1$ and so, that $v(\tau_{i/n}(\neg \phi))=1$. It means that $v(\neg \phi)\geq\frac{i}{n}$ or equivalently that $v(\phi)\leq 1-\frac{i}{n}<1$, a contradiction.

Proceed in a similar way to prove that the condition is sufficient.
\end{proof}

Note that since we have defined an accessibility relation for \emph{every} program $\alpha$ and the image of the valuation maps on  \emph{every} formula $\phi$, it is not clear that the canonical model is an $n+1$-valued (non-standard) model.  Indeed, in (non-standard) models, valuations are defined on atomic objects and inductively extended to all formulas. 

The canonical model will actually turn out to be an $n+1$-valued non standard model. The following lemma is a major step in the proof of this result. The proof of  this lemma is given in more general settings in \cite{Teheux2012}. We include a stand alone proof for the sake of readability.

Note that for any MV-homomorphism $u:\free_n \rightarrow \lucas_n$, the set $[\alpha]^{-1}u^{-1}(1)$ is a filter of $\free_n$ since the formula $\necess{\alpha}(\phi \implies \psi)\implies (\necess{\alpha}\phi \implies \necess{\alpha}\psi)$ belongs to $\logic{PDL}_n$ for any program $\alpha$ and any formulas $\phi$ and $\psi$.

\begin{lem}\label{lem:truth_lemma}
If $\phi \in \Form$, if $\alpha\in \Pi$ and if $u\in W^{c}$ then
\begin{equation}
\Val^{c}(u, \necess{\alpha}\phi)=\bigwedge\{\Val^{c}(v, \phi)\mid v \in R^{c}_{\alpha}u\}.
\end{equation}
\end{lem}
\begin{proof}
We have to prove that 
\begin{equation}
u(\necess{\alpha}\phi)=\bigwedge\{v(\phi) \mid v \in R_{\alpha}u\}.
\end{equation}
First, assume that $u(\necess{\alpha}\phi)=\frac{i}{n}$ for some $i \in \{1, \ldots, n\}$. It follows that 
\begin{equation}
1=\tau_{i/n}(u(\necess{\alpha}\phi))=u(\tau_{i/n}(\necess{\alpha}\phi))=u(\necess{\alpha}\tau_{i/n}(\phi)),
\end{equation}
where the first equality is obtained by definition of $\tau_{i/n}$, the second one holds because $u$ is an MV-homomorphism and the last one  from item (\ref{item:tau_respect}) of Proposition \ref{prop:tauto}. Hence, for any $v \in R_{\alpha}u$, we get $v(\tau_{i/n}(\phi))=\tau_{i/n}(v(\phi))=1$, which means that $v(\phi)\geq \frac{i}{n}$. We have proved that
\begin{equation}
u(\necess{\alpha}\phi)\leq\bigwedge\{v(\phi) \mid v \in R_{\alpha}u\}.
\end{equation}
For the other inequality, assume \emph{ad absurdum} that there is an $i\leq n$ such that
\begin{equation}
u(\necess{\alpha}\phi)< \frac{i}{n}\leq\bigwedge\{v(\phi) \mid v \in R_{\alpha}u\}, 
\end{equation}
\emph{i.e.}, such that $u(\necess{\alpha}\tau_{i/n}(\phi))\neq 1$ and $v(\tau_{i/n}(\phi))=1$ for any $v \in R_\alpha u$. Note that the definition of $R_\alpha$ means that the maximal filters above $[\alpha]^{-1}u^{-1}(1)$ are exactly the $v^{-1}(1)$ where $v$ belongs to $R_\alpha u$. Hence, the element $\tau_{i/n}(\phi)$ belongs to any maximal filter that contains  $[\alpha]^{-1}u^{-1}(1)$ but is not an element of $[\alpha]^{-1}u^{-1}(1)$, a contradiction.
\end{proof}

\begin{thm}
  The canonical model of $\logic{PDL}_n$ is an $n+1$-valued non standard \textsc{Kripke}
  model.
\end{thm}
\begin{proof}
 We prove the following properties of the canonical model.
\begin{enumerate}
\item $R_{[\alpha \cup \beta]}=R_\alpha \cup R_\beta$,
\item $R_{\alpha;\beta}=R_\alpha \circ R_\beta$,
\item $R_{\psi?}=\{(u,u) \mid \Val(u, \psi)=1\}$,
\item $R_{\alpha^*}$ is reflexive, transitive and contains $R_\alpha$,
\item $\model{M}\models \{[\alpha^*]\phi \leftrightarrow (\phi \wedge [\alpha][\alpha^*] \phi), (\phi \wedge [\alpha^*](\phi \implies [\alpha]\phi)^n) \implies [\alpha^*] \phi, [\alpha]^* \phi \implies [\alpha^*][\alpha^*]\phi\}$.
\end{enumerate}

For (1), we note that the inequality $R_\alpha \cup R_\beta\subseteq R_{\alpha \cup \beta}$ is trivial. For the other inequality, let us assume that $(u,v)$ belongs to $R_{\alpha \cup \beta}$ but not to $R_\alpha \cup R_\beta$. There are formulas $\phi$ and $\psi$  such that $\Val(u,\necess{\alpha} \phi)=1$,  $\Val(u,\necess{\beta} \psi)=1$ and $\Val(v, \phi \vee \psi)<1$. Then, thanks to Lemma \ref{lem:truth_lemma}, 
\begin{equation}
\Val(u, [\gamma] (\phi \vee \psi)) \geq \Val(u, [\gamma] \phi \vee [\gamma] \psi))=1,
\end{equation}
for $\gamma \in \{\alpha,\beta\}$.
Hence, 
\begin{equation}
\Val(u,\necess{\alpha \cup \beta}(\phi \vee \psi))=\Val(u, [\alpha](\phi \vee \psi) \wedge  [\beta](\phi \vee \psi) )=1
\end{equation}
while $\Val(v, \phi \vee \psi)<1$. We conclude that $(u,v)$ does not belong to $R_{\alpha \cup \beta}$, a contradiction.

The inequality $R_{\alpha}\circ R_{\beta} \subseteq R_{\alpha;\beta}$ of (2) is clear. Let us prove the other inequality. Assume that $(u,v)\in R_{\alpha;\beta}$. We prove that the filter generated by $[\alpha]^{-1}u^{-1}(1) \cup \poss{\beta}v^{-1}(1)$ is a proper filter of $\free_n$. Assume that $\phi_1, \ldots, \phi_k$ belong to $[\alpha]^{-1}u^{-1}(1)$, that $\psi_1, \ldots, \psi_l$ belong to $v^{-1}(1)$ and that $\epsilon_1, \ldots, \epsilon_k, \eta_1, \ldots, \eta_l$ are nonnegative integers. We prove that 
\begin{equation}
\Phi \odot \Psi \neq 0
\end{equation}
where $\Phi$ denotes the formula $\phi_1^{\epsilon_1} \odot \cdots \odot \phi_k^{\epsilon_k}$ and $\Psi$ the formula $(\poss{\beta}\psi_1)^{\eta_1} \odot \cdots \odot (\poss{\beta}\psi_l)^{\eta_l}$. 

Let us denote by $\Psi'$ the formula $\psi_1^{\eta_1}\odot \cdots \odot\psi_l^{\eta_l}$.
Since $(u, v)$ belongs to $R_{\alpha;\beta}$ and $v(\Psi')=1$ we obtain thanks to Lemma \ref{lem:alterative_def} that $u(\poss{\alpha;\beta}\Psi')=1$. It follows that $u(\necess{\alpha}\Phi \odot \poss{\alpha}\poss{\beta}\Psi')=1$ and hence, according to Lemma \ref{prop:tauto} (16)
that $u(\poss{\alpha}(\Phi \odot \poss{\beta} \Psi'))=1$. Then, according to Lemma \ref{lem:truth_lemma} and Remark \ref{rem:tauto_modal},
\begin{equation}
u(\poss{\alpha}(\Phi \odot \poss{\beta} \Psi'))=\bigvee\{w(\Phi \odot \poss{\beta} \Psi') \mid w \in R_\alpha u\}.
\end{equation}

 Hence, there is a $w$ in $R_\alpha u$ such that $w(\Phi \odot \poss{\beta} \Psi')=1$ which proves that $\Phi \odot \poss{\beta} \Psi'\neq 0$ in $\free_n$. It follows from Lemma \ref{prop:tauto} (\ref{tem:tauto})
 that $(\Phi \odot \poss{\beta} \Psi') \implies (\Phi \odot \Psi)$ is a theorem of $\logic{PDL_n}$ wich implies that 
\begin{equation}
\Phi \odot \Psi \geq \Phi \odot \poss{\beta} \Psi' >0
\end{equation}
in $\free_n$
 which is the desired conclusion.

(3) Thanks to axiom $[q?]p \leftrightarrow \neg q^n \vee p$ and the rule of uniform substitution, we obtain that $(u,v) \in R_{\psi?}$ if either $u(\psi)<1$ and $v(\phi)=1$ for any $\phi \in  \free_n$ (which is impossible since $v^{-1}(1)$ is a proper filter of $\free_n$) or $u(\psi)=1$ and $u^{-1}(1) \subseteq v^{-1}(1)$, which means that $v=u$ by maximality.

(4) Let us prove that $R_\alpha \subseteq R_{\alpha^*}$. Assume that $(u,v)\in R_\alpha$ and that $u([\alpha^ *]\phi)=1$ for some $\phi \in \free_n$. Thanks to the axioms that define the operator $*$, it means that $u(\phi \wedge [\alpha][\alpha^*]\phi)=1$. It follows that $u([\alpha][\alpha^*]\phi)=1$. Since $(u,v) \in R_\alpha$ we deduce that $v([\alpha^*]\phi)=1$, hence that $v(\phi \wedge [\alpha][\alpha^*]\phi)=1$ and finaly that $v(\phi)=1$.

Eventually, reflexivity and transitivity of $R_{\alpha^*}$ are easily obtained. 

(5) is obtained by construction.
\end{proof}

\begin{thm}\label{thm:main}
  The logic $\logic{PDL}_n$ is complete with respect to the $n+1$-valued
  \textsc{Kripke} models.
\end{thm}
\begin{proof}
  If $\phi$ is a tautology, then $\phi$ is valid in 
  $[\model{M}^{c}]_\phi$ which is an $n+1$-valued \textsc{Kripke} model. It follows from  Lemma
  \ref{lem:filtration} that $\phi$ is true in $\model{M}^{c}$. We thus conclude that $\phi$ is in any maximal filter of
  $\free_{n}$, \emph{i.e.} that $\phi \equiv 1$  and thus that $\phi$ is a theorem of $\logic{PDL}_n$.
\end{proof}

\section{Concluding remarks}\label{sec:concl}


This paper deals with some theoretical issues of a many-valued generalization of $\logic{PDL}$.
We believe that this generalization could reveal to be a valuable tool for analysis of problems arising from various fields such as computer science, epistemic logic or game theory.  We present a few ideas about possible areas in which this new language could be applied or generalized.
\subsection{Distributed algorithms}
 Some of the problems that can be solved  by a distributed or parallel algorithm  could be modeled with the language of $\logic{PDL}_n$. 

Consider as a toy example the problem of encoding a string $w$ of length $2n$ over an alphabet $\Sigma$ into an alphabet $\Sigma'$ using a coding function $c:\Sigma^2 \to \Sigma'^*$. Assume that this task is distributed over two processes $P_1$ and $P_2$ and that $P_1$ starts encoding from the head of the string and $P_2$ starts from its tail. 

Let us consider two propositional variables $p_1$ and $p_2$ that are evaluated at each step $u$ of the algorithm as $\Val(u,p_i)=k_i/n$ where for any $i\in\{1,2\}$, $k_i$ denotes the number of substrings of length $2$ that process $P_i$ has already encoded in step $u$. The algorithm can be modeled by a many-valued \textsc{Kripke} model. It terminates in step $u$ if $\Val(u, p_1\oplus p_2)=1$.

\subsection{Dynamic epistemic logic} In \cite{Benthem2006}, the authors show how to use the language of $\logic{PDL}$ to design a dynamic epistemic logic $\logic{LLC}$ for multi-agent systems that allows to deal with different kinds of information changes (public announcements, subgroup announcements, partial observations\ldots) or factual changes in the state of the world. In their settings, agents are represented by elements of $\Pi_0$ and the program operators '$;$', '$\cup$' and '$^*$' have epistemic interpretations (for example $[a;b]\phi$ is read `agent $a$ knows that agent $b$ knows $\phi$'). Their semantic uses two kind of models: \emph{epistemic models} (which are standard  \textsc{Kripke} models for $\logic{PDL}$) and \emph{update models} used to capture information changes. It also provides rules to update the former with the latter. The generalization of these constructions to a many-valued realm, using the language and the \textsc{Kripke} models introduced in this paper, could help to model situations involving partial or shared knowledge.

\subsection{Modal logic for games} Modal logic turned  out to be a valuable tool to study several kinds of game forms (\cite{Pauly2001}). For example, \textsc{Pauly} introduced in \cite{Pauly2002} a logic, called $\logic{CL}_N$ to reason about effective power in coalitional games (with set of players $N$). A set of outcome states $X$ is effective for a coalition $C\subseteq N$ if the players in $C$ can choose a joint strategy that leads to a state in $X$ no matter which strategies are adopted by the players not belonging to $C$.  $\logic{CL}_N$ is a multi-modal logic that is complete for a class of neighborhood models.  Some of the tools introduced in this paper could be used to set up a generalization of $\logic{CL}_N$ designed to capture the degree with which a coalition $C$ can encompass a fuzzy set of outcome states.


\appendix
\section{Proof of the Filtration Lemma}\label{sect:appendix}
\begin{defn}[\cite{Ladner1979}]\label{defn:closure}
Assume that $X$ is a set of formulas. The \emph{\textsc{Fisher-Ladner} closure} $\fish{X}$ of $X$ is the smallest subset $Y$ of $\Form$ such that
\begin{multicols}{2}
\begin{enumerate}
\item $X \subseteq Y$,
\item $\phi\in Y$ if $\neg\phi \in Y$,
\item $\{\phi,\psi\} \subseteq Y$ if $\phi \implies \psi \in Y$,
\item  $\phi \in Y$ if $[\alpha]\phi \in Y$,
\item $[\alpha][\beta]\phi \in Y$ if $[\alpha;\beta]\phi\in Y$,
\item $\{[\alpha]\phi, [\beta]\phi\}\!\subseteq\!Y$ if $[\alpha\cup\beta]\phi\in Y$,
\item $[\alpha][\alpha^*]\phi\in Y$ if $[\alpha^*]\phi \in Y$,
\item $\{\psi,\phi\} \subseteq Y$ if $[\psi?]\phi \in Y$.
\end{enumerate}
\end{multicols}
\end{defn}

\begin{lem}\label{lem:psiE}
If $\struc{W, R, \Val}$ is a weak  non-standard $n+1$-valued model, if $\phi \in \Form$ and if $E$ is a subset of $W$ which is $\equiv_\phi$-saturated  (\emph{i.e.}, $E$ contains $[u]_\phi$ whenever it contains $u$), then there is a formula $\Psi_E$ such that $E=\Val^{-1}(\cdot,\Psi_E)(1)$.
\end{lem}
\begin{proof}
For any $[t]\in [W]$, any $\rho\in \fish{\phi}$ and any $i\in \{0, \ldots, n\}$ such that $\Val([t], \rho)=\frac{i}{n}$, let us denote by $I_{\rho, [t]}$ the formula $I_{\frac{i}{n}}(\rho)$. Then, set
\begin{equation}
\psi_{[t]}=\bigwedge_{\rho \in \fish{\phi}} I_{\rho, [t]}.
\end{equation}
Then $u \in [t]$ if and only if $\Val(u, \psi_{[t]})=1$. The formula
\begin{equation}
\psi_E=\bigvee_{[t] \subseteq E} \psi_{[t]}.
\end{equation}
has the desired property.
\end{proof}

\begin{lem}\label{lem:lemlemfiltr}
Assume that $\model{M}$ is a non-standard \textsc{Kripke} model and that $\phi \in \Form$ and $\alpha \in \Pi$. If $\model{M}\models (\phi \implies \necess{\alpha} \phi)^n$ then $\model{M}\models \phi \implies \necess{\alpha^*}\phi$.
\end{lem}

\begin{proof}
Remark (\ref{defn:nonstand}) states that  axioms and rules of $\logic{PDL}_n$ are sound in non-standard \textsc{Kripke} frames. Lemma \ref{lem:IND} states that $\phi \implies \necess{\alpha^*}\phi$ can be obtained from $(\phi \implies \necess{\alpha} \phi)^n$ by applications of axioms and rules  of $\logic{PDL}_n$. If follows that if $\model{M}\models (\phi \implies \necess{\alpha} \phi)^n$ then $\model{M}\models \phi \implies \necess{\alpha^*}\phi$.
\end{proof}

We now provide the proof of the Filtration Lemma. 

\begin{proof}[Proof of Lemma \ref{lem:filtration}]
  The proofs of $(1)$ and $(2)$ are done by mutual induction.

(1) If $\psi \in \Prop$, the result follows directly from the definition of $[\model{M}]$. If $\psi=\rho \implies  \mu$ or $\psi=\neg \rho$, the result follows by applying induction hypothesis to $\rho$ and $\mu$.

If $\psi=[\alpha]\rho \in \fish{\phi}$ then $\rho \in \fish{\phi}$. 
We have to prove that 
\begin{equation}\label{eqn:mlmj}
\Val(u, \necess{\alpha}\rho)=\Val([u], \necess{\alpha}\rho).
\end{equation}
 First, we prove inequality $\leq$. We obtain successively
\begin{eqnarray}
\Val(u, [\alpha]\rho)  & \leq &\bigwedge\{\Val(v,\rho)\mid ([u], [v])\in R_\alpha\}\label{eqn:djsiz}\\
& = & \bigwedge\{\Val([v],\rho)\mid ([u], [v])\in R_\alpha\}\label{eqn:djsix}\\
& = & \Val([u],[\alpha]\rho),
\end{eqnarray}
where (\ref{eqn:djsiz}) and (\ref{eqn:djsix}) are obtained by (2) (b) and  by induction hypothesis for $\rho$.

Now, we prove inequality $\geq$ in (\ref{eqn:mlmj}). From (2)(a) we obtain that
$R_\alpha^{[\model{M}]}$ contains $([u],[v])$ whenever $(u,v)\in R_\alpha$. It follows that
\begin{eqnarray}
\Val([u],[\alpha]\rho) & = &  \bigwedge\{\Val([v],\rho)\mid ([u], [v])\in R_\alpha\}\\ 
 & = & \bigwedge \{\Val(v,\rho) \mid ([u],[v])\in R_\alpha\}\label{eqn:djsix01}\\
 & \leq &  \bigwedge \{\Val(v,\rho)\mid (u,v)\in R_\alpha\}\\
& = & \Val(u, [\alpha]\rho),
\end{eqnarray}
where (\ref{eqn:djsix01}) is obtained by induction hypothesis. Hence, we have proved (\ref{eqn:mlmj}).


(2)
There are five cases to consider according to the form of $\alpha$.

If $\alpha\in \Pi_0$ then, knowing that $[\alpha]\psi$ and $\psi$ are in $\fish{\phi}$, the result follows easily from the definition of $[\model{M}]$.

If $\alpha=\beta \cup \gamma$ then (a) is easily obtained by application of induction hypothesis to $[\alpha]\psi$ and $[\beta]\psi$ and the fact that $R_{[\beta\cup\gamma]}=R_\beta \cup R_\gamma$ in any (non standard) $n+1$-valued \textsc{Kripke} model. For (b), assume that $([u], [v])\in R_{[\beta\cup \gamma]}=R_{[\beta]} \cup R_{[\gamma]}$. We apply induction hypothesis to $[\beta]\psi$ and $[\gamma]\psi$ and we obtain that either $\Val(u,[\beta]\psi)\leq\Val(v,\psi)$ or  $\Val(u,[\gamma]\psi)\leq\Val(v,\psi)$. The result is then obtained thanks to item (\ref{tauto:dgdks}) of Proposition \ref{prop:tauto}.

If $\alpha=\beta;\gamma$, we can proceed in a similar way by application of induction hypothesis to $[\beta][\gamma] \psi, [\gamma] \psi \in \fish{\phi}$. 

If $\alpha=\rho?$ then $\rho \in \fish{\phi}$ and  we obtain by (1)  that $\Val(u, \rho)=\Val([u], \rho)$, which gives a proof of (a). For (b), we note that if $([u],[u])\in R_{\rho?}$ then $1=\Val([u],\rho)=\Val(u,\rho)$ thanks to (1) applied to $\rho\in\fish{\phi}$. It follows that
\begin{equation}
\Val(u,[\rho?]\psi)=\Val(u,\neg\rho^n \vee \psi)=\Val(u, \psi).
\end{equation}


If $\alpha=\beta^*$ then $[\beta][\beta^*]\psi \in \fish{\phi}$ and we can apply the induction hypothesis to $R_\beta$. To prove (a), assume that $(u,v)\in R_{\beta^*}$. Let us consider 
\begin{equation}
E=\{t \in W \mid ([u], [t]) \in R_{\beta^*}\}.
\end{equation}

The set $E$  is clearly $\equiv_\phi$-saturated. By Lemma \ref{lem:psiE}, there exists a formula $\Psi_E$ such that $E=\Val(\cdot, \Psi_E)^{-1}(1)$.

Since 
$R^{\scriptscriptstyle[ \model{M}]}_{\beta^*}$ is a reflexive extension of $R^{\scriptscriptstyle[\model{M}]}_\beta$, it follows that $u$ belongs to $E$. 

Now, assume that $s\in E$ and that $sR_\beta t$. By induction hypothesis, we obtain that $([s],[t])$ is in $R_\beta^{\scriptscriptstyle[\model{M}]}$. Then  $([u],[t])$ is in $R_{\beta^*}^{\scriptscriptstyle[\model{M}]}$ since this relation is a transitive extension of $R_{\beta}^{\scriptscriptstyle[\model{M}]}$. Hence
$
\model{M}\models (\Psi_E^n \implies [\beta] \Psi_E^n)^n.
$
By Lemma \ref{lem:lemlemfiltr} it follows that $\model{M} \models \Psi_E^n \implies [\beta^*]\Psi_E^n$. As $u \in E$, we conclude that $\Val(u, \Psi_E^n)=1$ so that $\Val(u, [\beta^*]\Psi_E^n)=1$ and $\Val(v, \Psi_E^n)=1$ since $(u,v)\in R_{\beta^*}$. Thus, we have proved that  $([u],[v])\in R_{\beta^*}^{[\model{M}]}$.

To prove (b), assume that $([u],[v])\in R_{\beta^*}$. Then, since $R_{\beta^*}^{[\model{M}]}$ is the reflexive and transitive closure of $R_{\beta}$, there are some $[w_i]$ ($i \in \{0, \ldots, m+1\}$) in $[\model{M}]$ such that $([w_i], [w_{i+1}]) \in R_{\beta}$ for any $i\leq m$ and such that $[u]=[w_0]$ and $[v]=[w_{m+1}]$. We obtain by induction hypothesis that for any $i\leq m$,
\begin{equation}
\Val(w_i,[\beta^*]\psi)\leq\Val(w_i, \psi \wedge [\beta][\beta^*]\psi)\leq \Val(w_{i+1},[\beta^*]\psi).
\end{equation}
Hence, we eventually obtain that $\Val(u,[\beta^*]\psi)\leq  \Val(v, \psi)$.
\end{proof}

\bibliographystyle{plain}
\bibliography{BiblioMPDL}
\end{document}